\documentclass{article}

\usepackage{PRIMEarxiv}

\usepackage[utf8]{inputenc} 
\usepackage[T1]{fontenc}    
\usepackage{hyperref}       
\usepackage{url}            
\usepackage{booktabs}       
\usepackage{amsfonts}       
\usepackage{nicefrac}       
\usepackage{microtype}      
\usepackage{lipsum}
\usepackage{fancyhdr}       
\usepackage{graphicx}       
\graphicspath{{media/}}     
\usepackage{amsmath}

\usepackage{amsthm}
\usepackage{amssymb}

\usepackage{wrapfig}
\usepackage{lscape}
\usepackage{rotating}

\providecommand{\tabularnewline}{\\}

\newtheorem{definition}{Definition}{\bfseries}{\itshape}
\newtheorem{proposition}{Proposition}{\bfseries}{\itshape}
\newtheorem{theorem}{Theorem}{\bfseries}{\itshape}

\title{Conciliating Privacy and Utility in Data Releases via
	Individual Differential Privacy and Microaggregation
}

\author{
  Jordi Soria-Comas, David S\'anchez, Josep Domingo-Ferrer, Sergio Mart\'{\i}nez, Luis Del Vasto-Terrientes \\
  Department of Computer Engineering and Mathematics,\\
  UNESCO Chair in Data Privacy,\\ CYBERCAT-Center for Cybersecurity Research of Catalonia, \\
  Universitat Rovira i Virgili, \\
  Av. Paisos Catalans 26, Tarragona, 43007, Catalonia\\
  \texttt{\{jordi.soria, david.sanchez, josep.domingo, sergio.martinezl\}urv@cat} \\
  luis.delvasto@outlook.com
   \And
}

\begin{document}
\maketitle

\begin{abstract}
$\epsilon$-Differential privacy (DP) is a 
well-known privacy model that offers strong privacy guarantees. However, when applied to data releases, DP significantly deteriorates the 
analytical utility of the protected outcomes. To keep data utility at reasonable levels, practical
applications of DP to data releases have used weak privacy parameters (large $\epsilon$), which 
dilute the privacy guarantees of DP. In this work, we tackle this issue by
using an alternative formulation of the DP privacy 
guarantees, named $\epsilon$-individual differential privacy (iDP), which 
causes less data distortion while providing the same protection as DP to subjects.
We enforce iDP in data releases by relying on attribute masking plus 
a pre-processing step based on data microaggregation. The goal 
of this step is to reduce the sensitivity to record changes, 
which determines the amount of noise required to enforce iDP (and DP). 
Specifically, we propose data microaggregation strategies designed 
for iDP whose sensitivities are significantly lower than those 
used in DP. As a result, we obtain iDP-protected data with significantly better utility than with DP. 
We report on experiments that show how our approach can provide strong privacy 
(small $\epsilon$) while yielding protected data that do not significantly degrade the accuracy of secondary data analysis. 
\end{abstract}


\section{Introduction}
\label{intro}
Data analysis has become an essential tool in today's world. Its applications 
 range from 
the enhancement of customers' experience ({\em e.g.} via recommender systems) to the support of strategic
decision making ({\em e.g.} using data mining), and in general may substantially
improve human life and human endeavors. 
 However, when using data on people for secondary purposes, the privacy
of the subjects in the data set must be preserved. This is increasingly important under the new stronger privacy regulations, epitomized by 
the European General Data Protection Regulation (GDPR).

Personally identifiable information (PII) should be protected
before releasing or sharing it for analysis. 
Several approaches are possible: one may release
a fixed set of statistics, offer interactive access to the PII 
via some query mechanism that provides anonymized answers, or release 
an anonymized data set. The latter is the most convenient way for the data analyst, and the one that we consider 
in this work~\cite{SDC}. Releasing sufficiently 
useful protected data sets gives freedom to 
the data analyst to carry out unconstrained exploratory 
data analyses, data mining tasks or even machine learning at will. 

We use the term microdata set to refer to a data set whose records contain detailed information about
a single subject. Privacy protection in microdata releases is a discipline with a long history. 
It was initially developed 
in the context of official statistics 
under the name of Statistical Disclosure Control (SDC)~\cite{SDC}. At that time,
the number of data controllers was limited, and the data were collected
under a strong pledge of privacy. This allowed making reasonable assumptions
about the side knowledge available to intruders for them to conduct inference
attacks leading to disclosure. Such assumptions were very useful to
adjust the SDC methods in view of obtaining adequate disclosure protection.
However, with
the development of IT, the landscape changed radically. Nowadays,
large amounts of heterogeneous personal data are collected by a large number
of data controllers~\cite{Soria2015big}, public and private; in this context, making
assumptions about the knowledge available to intruders is quite difficult,
and using robust privacy models to protect data releases seems the best option.

$\epsilon$-Differential privacy~\cite{DworkMNS06} ($\epsilon$-DP) is a well-known privacy
model whose privacy guarantees are independent of the side knowledge available
to intruders. This makes DP particularly suitable in the current landscape.
However, unlike other privacy models originally developed
for data releases (\emph{e.g.} $k$-anonymity~\cite{Samarati2001}, 
$l$-diversity~\cite{Machanavajjhala2007} or $t$-closeness~\cite{Li2007}),
DP was initially designed for interactive settings, in which only 
the outcomes of queries submitted to a database held by a trusted 
party are protected. 

More recently, DP has also been used in the more convenient non-interactive setting, in which the aim is to release anonymized data sets~\cite{Zhang2014,Xiao2010,Soria2014,Sanchez2016,Chu2023}
Yet, in such a setting, DP introduces large distortions in the protected outcomes thereby significantly hampering their analytical utility.
As a result, DP has only being deployed to a limited extent in real-world 
applications and,
when done, the privacy requirements ($\epsilon$ value) have been severely 
relaxed in order to keep data reasonably useful~\cite{limits}.
A paradigmatic example is the recent use of DP by the U.S. Census Bureau to protect the 2020 Decennial Census release~\cite{Abowd}. To retain some utility, they were
forced to take $\epsilon=39.9$~\cite{USCB22} and, even with this large value, data utility significantly degraded w.r.t. the former Census releases using non-DP data protection~\cite{Kenny21}. In this sense, it is well-known that employing $\epsilon$ values larger than 1 dilutes the DP privacy guarantees until the point that \emph{DP delivers privacy mostly in name}~\cite{Dwork19}. 

In order to reconcile data utility with DP-like privacy guarantees, we proposed $\epsilon$-individual differential privacy~\cite{iDP}
($\epsilon$-iDP), a privacy model that can incur less information
loss than the standard DP, while giving subjects the same privacy protection as 
DP. The focus of this work is to design mechanisms to use iDP in data releases and thereby 
benefit from its enhanced utility-privacy trade-off.

\subsection*{Contribution and Plan of This Paper} 

To take advantage of iDP in data releases, we propose several strategies based on data microaggregation~\cite{Domi02} whose local sensitivities are significantly lower than the
global sensitivity required in standard DP. 
In this way, we enable privacy-preserving data releases that offer
the robust privacy guarantees of DP at the individual level while preserving data utility 
significantly better than standard DP.

Due to its definition, the microaggregation-based iDP-protected data we obtain can never offer less protection than
the underlying microaggregation, which hides individuals in a group. In fact, our work shows that it does much better: {\em our approach allows using the small values of $\epsilon$ recommended in~\cite{Dwork2011}} 
(the only ones that are actually meaningful in DP-like privacy models) {\em while maintaining analytical utility to an extent that standard DP cannot offer for such small values}.

We validate the above through a set of experiments on several standard data sets, whose utility is evaluated through general-purpose utility metrics and in machine learning tasks. 

The rest of this paper is organized as follows. In Section~\ref{subsec:dp} we give
background on DP and iDP. 
In Section~\ref{sec:process}, we review our approach to DP data releases and detail our proposal
to generate iDP data sets through carefully tailored microaggregation strategies.
In Section~\ref{sec:experiments},
we report on the experiments we conducted on several data sets.
Section~\ref{sec:Conclusions} gathers conclusions and 
identifies future research lines.

\section{From Differential Privacy to Individual Differential Privacy\label{subsec:dp}}

Differential privacy~\cite{DworkMNS06} stands out because of the strong privacy
guarantees it offers. DP does not make any assumptions about the
side knowledge available to the intruders; rather, disclosure risk
limitation is tackled in a relative manner: the result of any analysis
should be similar between data sets that differ in one record. 
Assuming 
that each record corresponds to an individual, the rationale of DP 
is explained in~\cite{Dwork2006}:
\begin{quotation}
	Any given disclosure will be, within a multiplicative factor, just
	as likely whether or not the individual participates in the database.
	As a consequence, there is a nominally higher risk to the individual
	in participating, and only nominal gain to be had by concealing or
	misrepresenting one's data.
\end{quotation}
With DP, individuals should not be reluctant to participate in the data
set because the risk of disclosure is only very marginally increased
by such participation.

Differential privacy assumes a trusted party that:
(i) holds the database, (ii) receives the queries submitted by the
data users, and (iii) responds to them in a privacy-aware manner.
The notion of differential privacy is formalized according to the
following definition.

\begin{definition}[$\epsilon$-Differential privacy]
	\label{def:dp}A randomized function $\kappa$ gives $\epsilon$-dif\-fer\-ent\-ial
	privacy if, for all data sets $D_{1}$ and $D_{2}$ that differ in
	one record (\emph{a.k.a.} neighbor data sets), and all $S\subset Range(\kappa)$,
	we have 
	\[
	\Pr(\kappa(D_{1})\in S)\text{\ensuremath{\le}}\exp(\epsilon)\Pr(\kappa(D_{2})\in S).
	\]
\end{definition}

Safe values for $\epsilon$ are 0.01, 0.1~\cite{Dwork2011};
for larger values, the privacy guarantees of DP tend to vanish 
because the DP-protected outcomes may be substantially 
affected by the presence or absence of each individual, 
which increases the disclosure risk~\cite{iDP}.  

DP has composability properties, that is, aggregating several differentially private results still satisfies
DP although, sometimes, with a different $\epsilon$. 

\begin{theorem}[Sequential composition]
	\label{teo1} Let $\kappa_{1}$ be a randomized function giving $\epsilon_{1}$-DP
	and $\kappa_{2}$ a randomized function giving $\epsilon_{2}$-DP.
	Then, any deterministic function of $(\kappa_{1},\kappa_{2})$ gives
	$(\epsilon_{1}+\epsilon_{2})$-DP. 
\end{theorem}

\begin{theorem}[Parallel composition]
	\label{teo2} Let $\kappa_{1}$ and $\kappa_{2}$ be randomized functions
	giving $\epsilon$-DP. If $\kappa_{1}$ and $\kappa_{2}$ are applied
	to {\em disjoint} data sets or subsets of records, any deterministic
	function of $(\kappa_{1},\kappa_{2})$ gives $\epsilon$-DP. 
\end{theorem}

For a numerical query $f$, $\epsilon$-DP can be attained via noise
addition; that is, by adding some random noise to the actual query
result: $\ensuremath{\kappa(x)=f(x)+N}$.
The amount of noise that needs to be added depends on the variability
of the query function between neighbor data sets, that is, 
on the global sensitivity of the query.

\begin{definition}[Global sensitivity]
	\label{def:global_sensitivity}Let $f$ be a function that is evaluated at data
	sets
	in $\mathcal{D}$
	and returns values in $\mathbb{R}^{k}$. The global sensitivity of $f$ 
	in $\mathcal{D}$ is 
	\[
	\Delta f=\max_{\begin{array}{c}
		{\scriptstyle {D_1,D_2\in\mathcal{D}}}\\
		{\scriptstyle {d(D_1,D_2)=1}}
		\end{array}}\left\Vert f(D_1)-f(D_2)\right\Vert _{1},
	\]
	where $d(D_1,D_2)$ means that data sets $D_1$ and $D_2$ differ in one record. 
\end{definition}

Even though several noise distributions are possible, the Laplace
distribution~\cite{DworkMNS06} is the most commonly employed one.

As mentioned in Section~\ref{intro}, the deployment of DP has been
limited in practice, in spite of its strong privacy guarantees.
In fact, those guarantees are only meaningful for very small $\epsilon$ values,
but practitioners need to use large, unsafe $\epsilon$ values to preserve
sufficient utility.
In an attempt to improve the accuracy of the protected data, several relaxations of DP have
been proposed, such as $(\epsilon,\delta)$-DP~\cite{DworkKMMN06}, concentrated DP~\cite{DworkR16}, zero-concentrated DP~\cite{BunandSteinke2016} or R\'enyi DP~\cite{Mironov2017}. 
Essentially, these relaxations allow DP guarantees to be broken either by a small amount or
with a small probability. An alternative relaxation of DP is iDP~\cite{iDP}. 
The latter is particularly interesting because, unlike
the above, it preserves the privacy guarantees that DP gives to individual 
subjects for a given $\epsilon$.

Next, we recall the rationale of iDP.
When formalizing DP in Definition~\ref{def:dp}, the trusted party
is not allowed to take advantage of her knowledge about the actual
data set to adjust the level of noise. This leads to the formalization of DP 
being stricter than required by the intuition of DP (see quotation above),
which results in unnecessary accuracy loss. Let us explain this in greater
detail.

Consider an individual subject $I$ who has to decide between participating
in a data set or not. To neutralize any reluctance by $I$ to disclose
her private information, $I$ is told that query answers based on
the data set will not allow anyone to learn anything that was not
learnable without $I$'s presence; this is precisely the intuitive
privacy guarantee DP offers. To attain such privacy guarantees, DP
requires the response to be indistinguishable between any pair of
neighbor data sets. While such a requirement yields the target privacy
guarantees, it is an overkill because the trusted party is not allowed
to take advantage of her knowledge of the data set. In other words,
if $D$ is the collected data set, the target privacy guarantees can
be attained by just requiring indistinguishability of the responses
between $D$ and its neighbor data sets. Notice that, although the
data set $D$ is not known until all the individuals have made their
decisions about participating/contributing to it, it is known to the
trusted party at the time of query response. 

According to the previous discussion, $\epsilon$-individual differential
privacy ($\epsilon$-iDP) is defined as follows.

\begin{definition}[$\epsilon$-Individual differential privacy~\cite{iDP}]
	\label{def:idp}Given a data set $D$, a response mechanism $\kappa(\cdot)$
	satisfies $\epsilon$-individual differential privacy (or $\epsilon$-iDP)
	if, for any neighbor data set $D'$ of $D$, and any $S\subset Range(\kappa)$
	we have 
	\[
	\begin{array}{l}
	\exp(-\epsilon)\Pr(\kappa(D')\in S)\le\Pr(\kappa(D)\in S)\\
	\text{\ensuremath{\le}}\exp(\epsilon)\Pr(\kappa(D')\in S).
	\end{array}
	\]
\end{definition}

In line with Definition~\ref{def:dp}, iDP requires the probability
of any result to differ between neighbor data sets at most by a factor
$\exp(\epsilon)$. However, unlike in Definition~\ref{def:dp},
the role of the data sets $D$ and $D'$ is not exchangeable: $D$
refers to the actual data set, and $D'$ to a neighbor data set of
$D$. The asymmetry between $D$ and $D'$ is relevant, because indistinguishability
is achieved only between $D$ and its neighbor data sets. As a side
effect of this asymmetry, we need to explicitly enforce an upper bound
($\Pr(\kappa(D)\in S)\text{\ensuremath{\le}}\exp(\epsilon)\Pr(\kappa(D')\in S)$)
and a lower bound ($\exp(-\epsilon)\Pr(\kappa(D')\in S))\le\Pr(\kappa(D)\in S)$).
This was not needed in Definition~\ref{def:dp} because the upper
bound could be obtained from the lower bound by exchanging the roles
of $D$ and $D'$.

This difference has an important and beneficial practical consequence.
Unlike DP, which requires
calibrating the added noise to the global sensitivity (that is, to the greatest change
between any pair of neighbor data sets), iDP can be attained 
by calibrating noise to local
sensitivity 
(which is normally much lower).
\begin{definition}[Local sensitivity~\cite{Nissim2007}]
	\label{def:local}
	The local sensitivity
	of a query function $f$ at a data set $D$ is 
	\[
	LS_{f}(D)=\max_{y:d(y,D)=1}\left\Vert f(y)-f(D)\right\Vert _{1},
	\]
	where $d(y,D)$ means that data set $y$ differs from $D$ in one record. 
\end{definition}

The following result is proven in~\cite{iDP}.

\begin{proposition}
	\label{prop:iDP_Lap}Let $f$ be a query function that takes values
	in $\mathbb{R}^{k}$. The mechanism $\kappa(x)=f(x)+(N_{1},\ldots,N_{k})$,
	where $N_{i}$ are independent identically distributed $Laplace(0,LS_{f}(D)/\epsilon)$
	random noises, gives $\epsilon$-iDP.
\end{proposition}

\section{iDP Data Sets via Individual Ranking Microaggregation\label{sec:process}}

Even though the natural application of DP and iDP is the interactive setting, both can be used to generate protected data sets via data masking. In the following, 
we first discuss how this can be achieved for standard DP, and then we tailor masking to reap the utility-preserving advantages of iDP.

\subsection{DP Data Sets\label{subsec:ir_dp}}

For years, the usual approach to generate DP data sets was based on computing 
DP histograms~\cite{Xiao2010,Zhang2014}; that is, on approximating
the data distribution by partitioning the data domain 
and counting the number of records in each partition set in a DP manner.
However, histogram-based
approaches have severe limitations when the number of attributes grows:
for a fixed granularity in each attribute, the number of histogram
bins grows exponentially with the number of attributes, which has
a devastating effect on both computational cost and accuracy. To mitigate
these issues, an alternative approach has been proposed 
that is based on masking attribute values of the
records in the original data set~\cite{Soria2014,Sanchez2016}. 
In order to reduce the sensitivity to record changes (which is the factor
that basically determines the noise needed to attain DP and hence
the utility damage incurred),
this approach applies a microaggregation step before masking.

Let $D$ be the collected data set. Assume that one wants to generate
$D_{\epsilon}$ --an ano\-nym\-iz\-ed version of $D$-- that satisfies
$\epsilon$-DP. Let $I_{r}(D)$ be the query that returns $r$. Then, one
can think of the data set $D$ as the collected answers to the queries
$I_{r}(D)$ for $r\in D$, and one can generate $D_{\epsilon}$ by
collecting $\epsilon$-DP responses to the previous queries. However,
since the purpose of DP is to make sure that individual records
do not have any significant effect on query responses, the amount 
of noise that should be added to the responses of $I_{r}(D)$
to fulfill DP is necessarily high, which severely deteriorates the accuracy 
of $D_{\epsilon}$.

To make masking a viable option for generating DP data sets,  
 the sensitivity of individual records needs to be reduced. This was attained
in~\cite{Sanchez2016,Soria2017mdai} for standard DP by adding a microaggregation
step before the actual DP data set generation. Even though microaggregation was 
initially proposed as an anonymization technique
in its own right~\cite{Domi02}, in this work we will use it 
as a means to reduce the
sensitivity of the queries.

Microaggregation proceeds in two steps:
\begin{enumerate}
	\item Split the data set into clusters of similar records of cardinality greater 
	than or equal to $k$ (a given parameter).
	\item Compute a representative record of each cluster and replace each record
	in the cluster by a copy of the representative.
\end{enumerate}

In Figure~\ref{fig:micro} we describe the generation of a 
microaggregated data set $\bar{D}$
via univariate microaggregation of each attribute in $D$. Let $D$
contain data about attributes $A^{1},\ldots,A^{m}$, for individuals
$I_{1},\ldots,I_{n}$. To microaggregate attribute $A^{a}$,
we cluster records by their similarity w.r.t. attribute $A^{a}$, compute
the centroid for $A^{a}$ of each cluster, and replace the original values in
$A^{a}$ by the corresponding centroid. Formally, 
let $C^{a}=\{C_{j}^{a}\}_{j}$ be the
clustering associated with attribute $A^{a}$, let $c_{j}^{a}$ be the
centroid associated with cluster $C_{j}^{a}$, and let $\rho_{a}(I_{i})$
be the index of the centroid associated with individual $I_{i}$. We
replace $x_{i}^{a}$ (the value of attribute $A^{a}$ for individual
$I_{i}$) by the corresponding centroid $c_{\rho_{a}(I_{i})}^{a}$.

\begin{figure}
	\begin{centering}
		\begin{tabular}{cccccccccc}
			&  & $D$ &  &  &  & $\bar{D}$ &  &  & \tabularnewline
			\cline{2-4} \cline{6-8} 
			& $A^{1}$ & $\ldots$ & $A^{m}$ &  & $A^{1}$ & $\ldots$ & $A^{m}$ &  & \tabularnewline
			\cline{2-4} \cline{6-8} 
			$I_{1}$ & $x_{1}^{1}$ & $\ldots$ & $x_{1}^{m}$ &  & $c_{\rho_{1}(I_{1})}^{1}$ & $\ldots$ & $c_{\rho_{m}(I_{1})}^{m}$ &  & \tabularnewline
			$I_{2}$ & $x_{2}^{1}$ & $\ldots$ & $x_{2}^{m}$ & $\rightarrow$ & $c_{\rho_{1}(I_{2})}^{1}$ & $\ldots$ & $c_{\rho_{m}(I_{2})}^{m}$ &  & \tabularnewline
			$\vdots$ & $\vdots$ &  & $\vdots$ &  & $\vdots$ &  & $\vdots$ &  & \tabularnewline
			$I_{n}$ & $x_{n}^{1}$ & $\ldots$ & $x_{n}^{m}$ &  & $c_{\rho_{1}(I_{n})}^{1}$ & $\ldots$ & $c_{\rho_{m}(I_{n})}^{m}$ &  & \tabularnewline
			\cline{2-4} \cline{6-8} 
		\end{tabular}
		\par\end{centering}
	\caption{Generation of $\bar{D}$ via univariate microaggregation of each attribute
		in $D$\label{fig:micro}}
	
\end{figure}

Once the microaggregated data set $\bar{D}$ has been created, we
generate $\bar{D}_{\epsilon}$, a DP version of $\bar{D}$, through the process
depicted in Figure~\ref{fig:dp}. 
\begin{itemize}
	\item 
	We work independently with each attribute $A^{a}$ to make it $\epsilon_{a}$-DP, where $\epsilon_{a}$
	is the share of the privacy budget assigned to attribute $A^{a}$.
	Afterwards, we combine all the DP attributes to generate the $\epsilon$-DP data set. By sequential 
	composition (see Theorem~\ref{teo2}), the overall privacy budget must be split among each of the attributes:
	$\epsilon=\sum\epsilon_{a}$.
	\item In the generation of the $\epsilon_{a}$-DP attributes, parallel composition applies
	(see Theorem~\ref{teo1}) because each centroid depends on a disjoint set of individuals (records).
	Thus, we can use the entire privacy budget assigned to attribute $A^a$, $\epsilon_a$, to mask
	each of the centroids associated with $A^a$. This is done by masking each centroid $c_{j}^{a}$ 
	(\emph{e.g.}, via Laplace noise) according to its global
	sensitivity ($\Delta c_{j}^{a})$ and privacy budget $\epsilon_{a}$.
\end{itemize}

The global sensitivity of a centroid is the maximum change in the centroid value
that can ensue from a change in one of the records in the cluster. Since the sensitivity
of the centroids is smaller than the sensitivity of the original records,
making $\bar{D}$ differentially private requires less noise than making $D$ differentially private. Loosely
speaking, centroids are less sensitive than individual
records because the former are an aggregation of several records. 
Counterintuitively,
even though the prior microaggregation step distorts data to some extent,
by starting from the microaggregated data set rather than the original
data set, we manage to obtain a DP data set that preserves substantially
more utility for a given $\epsilon$, as shown in~\cite{Sanchez2016}.
This is so because the noise reduction enabled by the prior microaggregation
(due to cluster centroids being less sensitive than individual records),
more than compensates the information loss introduced by such microaggregation;
specifically, whereas noise is random, microaggregation can exploit the underlying structure of data.

\begin{figure*}
	\begin{centering}
		\begin{tabular}{cccccccc}
			&  & $\bar{D}$ &  &  &  & $\bar{D}$ & \tabularnewline
			\cline{2-4} \cline{6-8} 
			& $A^{1}$ & $\ldots$ & $A^{m}$ &  & $A^{1}$ & $\ldots$ & $A^{m}$\tabularnewline
			\cline{2-4} \cline{6-8} 
			$I_{1}$ & $c_{\rho_{1}(I_{1})}^{1}$ & $\ldots$ & $c_{\rho_{m}(I_{1})}^{1}$ &  & $c_{\rho_{1}(I_{1})}^{1}+n_{\rho_{1}(I_{1})}^{1}$ & $\ldots$ & $c_{\rho_{m}(I_{1})}^{m}+n_{\rho_{m}(I_{1})}^{m}$\tabularnewline
			$I_{2}$ & $c_{\rho_{1}(I_{2})}^{1}$ & $\ldots$ & $c_{\rho_{m}(I_{2})}^{1}$ & $\rightarrow$ & $c_{\rho_{1}(I_{2})}^{1}+n_{\rho_{1}(I_{2})}^{1}$ & $\ldots$ & $c_{\rho_{m}(I_{2})}^{m}+n_{\rho_{m}(I_{2})}^{m}$\tabularnewline
			$\vdots$ & $\vdots$ &  & $\vdots$ &  & $\vdots$ &  & $\vdots$\tabularnewline
			$I_{n}$ & $c_{\rho_{1}(I_{n})}^{1}$ & $\ldots$ & $c_{\rho_{m}(I_{n})}^{1}$ &  & $c_{\rho_{1}(I_{n})}^{1}+n_{\rho_{1}(I_{n})}^{1}$ & $\ldots$ & $c_{\rho_{m}(I_{n})}^{m}+n_{\rho_{m}(I_{n})}^{m}$\tabularnewline
			\cline{2-4} \cline{6-8} 
			\multicolumn{8}{c}{}\tabularnewline
			\multicolumn{8}{c}{where $n_{j}^{a}$ is drawn from $Laplace(0,\Delta c_{j}^{a}/\epsilon_{a})$}\tabularnewline
		\end{tabular}\caption{Generation of $\bar{D_{\epsilon}}$ by masking the centroids with
			the appropriate amount of noise $D$\label{fig:dp}}
		\par\end{centering}
\end{figure*}

Algorithm~\ref{alg:ir_dp} formalizes the process described in Figures~\ref{fig:micro}
and~\ref{fig:dp}. The algorithm receives the original data set,
the privacy budget assigned to each attribute and the microaggregation
algorithm associated with each attribute. For each attribute, we run
the microaggregation (line 08), compute the global sensitivity of
the centroids (line 10), draw the noise from the Laplace distribution
(line 11), and mask each occurrence of the centroid (line 12). The algorithm
receives the microaggregation algorithms as a parameter. 
As discussed in~\cite{Sanchez2016}, a sensible
choice is to use individual-ranking microaggregation (form clusters
containing consecutive values of the attribute) and compute the centroid
as the arithmetic mean of the values in the cluster: 
\[
c_{j}^{a}=\frac{1}{|C_{j}^{a}|}\sum_{x\in C_{j}^{a}}x.
\]
In this case, the global sensitivity of a
centroid equals the size of the attribute domain over the size of
the associated cluster:
\[
\Delta c_{j}^{a}=\frac{\max A^{a}-\min A^{a}}{|C_{j}^{a}|},
\]
where $\max A^{a}$ is the maximum of the domain of $A^{a}$, and
$\min A^{a}$ is the minimum of the domain of $A^{a}$.

\begin{algorithm}
	\caption{\label{alg:ir_dp}$\epsilon$-DP data set generation via univariate
		microaggregation and Laplace noise addition}
	
	01\hspace{0.2cm}\textbf{Require}:
	
	02\hspace{0.2cm}\hspace{0.5cm}$D:$ data set with attributes $A^{1},\ldots,A^{m}$
	
	03\hspace{0.2cm}\hspace{0.5cm}$\epsilon_{a}:$ privacy budget assigned
	to $A^{a}$ (with $\epsilon=\sum\epsilon_{a})$
	
	04\hspace{0.2cm}\hspace{0.5cm}$M_{a}$: univariate microaggregation
	algorithm over attribute $A^{a}$
	
	05\hspace{0.2cm}\textbf{Output}:
	
	06\hspace{0.2cm}\hspace{0.5cm}$\bar{D}_{\epsilon}$: $\epsilon$-DP
	data set 
	
	\vspace{0.5cm}
	
	07\hspace{0.2cm}\textbf{for} $a=1$ \textbf{to} $m$
	
	08\hspace{0.2cm}\hspace{0.5cm}\textbf{let} $(C_{j}^{a},c_{j}^{a})_{j}$
	= the clusters and centroids produced by applying $M_{a}$ to $A^{a}$ 
	
	09\hspace{0.2cm}\hspace{0.5cm}\textbf{for} \textbf{each} $C_{j}^{a}$
	
	10\hspace{0.2cm}\hspace{0.5cm}\hspace{0.5cm}\textbf{let} $\Delta c_{j}^{a}$
	be the global sensitivity of cluster $c_{j}^{a}$
	
	11\hspace{0.2cm}\hspace{0.5cm}\hspace{0.5cm}\textbf{let} $n_{j}^{a}$
	be a draw from a $Laplace(0,\Delta c_{j}^{a}/\epsilon_{a})$ distribution
	
	12\hspace{0.2cm}\hspace{0.5cm}\hspace{0.5cm}Replace each $c_j^a\in C_{j}^{a}$
	by $c_{j}^{a}+n_{j}^{a}$

	13\hspace{0.2cm}\hspace{0.5cm}\textbf{end} \textbf{for}
	
	14\hspace{0.2cm}\textbf{end} \textbf{for}
	
	15\hspace{0.2cm}\textbf{return} $D$
\end{algorithm}

\subsection{iDP Data Sets}

To enforce DP, the proposal described in Algorithm~\ref{alg:ir_dp}
needs to use the global sensitivity of the attributes.
Unfortunately, global sensitivities can be very large, 
especially when attributes have domains that are much larger than 
their actual value ranges in the original data set. This results in much 
noise being added, which severely damages the accuracy of the generated data set.
Also, it may be difficult to 
bound the domain of some attributes 
(what is the upper bound of an attribute such as \emph{income}?) 
or it may be thoroughly impossible (if the attribute is naturally unbounded).
In such cases, attribute domains should be artificially limited to reasonably large
bounds, which will also significantly increase the global sensitivity.
These issues are the result of mechanisms enforcing DP
not being allowed to leverage the knowledge of $D$, due to the strict
formulation of DP (see Section~\ref{subsec:dp}).

In the following, we face these issues with a proposal 
to generate iDP data sets. Under 
iDP, individual subjects are given the same privacy guarantees  
as under DP (\emph{i.e.}, the presence
or absence of any single subject's data does not have a significant effect on the
protected data set), but the accuracy of the iDP data
set is better because local sensitivities are used rather than global ones. 

First, we describe
a naive adaptation of the approach described above for DP to iDP: the global
sensitivity is merely replaced by the local sensitivity. Then, we take several steps to
use the knowledge of $D$ and to further reduce
the sensitivity. This is done by introducing a pre-processing step
that alters the data set before microaggregation.

\subsubsection{iDP Data Sets Using Local Sensitivity\label{subsec:first}}

By using iDP, we can adjust the noise to the local sensitivity of the attributes
rather than to their global sensitivity: 
\begin{itemize}
	\item The local sensitivity of an attribute centroid is the maximum change in the centroid
	value that can occur when switching from $D$ to a neighbor data set.
	In other words, it is the maximum change of the centroid value that can result 
	from a change in
	one of the records in $D$. In the worst case, we may change the smallest
	record in the cluster w.r.t. attribute $A$ to $\max A$ (the maximum value of $A$'s domain),
	or change the largest record in the cluster to $\min A$ (the minimum value
	of $A$'s domain). 
	\item On the other hand, the global sensitivity of a centroid 
	is the maximum of the local sensitivity
	across all pairs of neighbor data sets. Thus, in general, the global sensitivity
	is greater than the local sensitivity. They are equal only if the
	original data set is 
	one in which the distance between the smallest 
	(largest) record to $\max A$ ($\min A$) equals
	the distance between $\max A$ and $\min A$.
\end{itemize}

The first mechanism we propose is an adaptation of Algorithm~\ref{alg:ir_dp}
to iDP: we merely replace 
the global sensitivity ($\Delta c_{j}^{a})$
by the local sensitivity at $\bar{D}$, 
which is computed for each cluster centroid ($LS_{c_{j}^{a}}(\bar{D}))$.
Notice that, since each centroid is computed on a disjoint set of records, by
parallel composition we can work with each centroid independently.
More specifically, the changes to Algorithm~\ref{alg:ir_dp} are:
\begin{itemize}
	\item At line 10, we compute $LS_{c_{j}^{a}}(\bar{D})$, the local sensitivity
	of the centroid $c_{j}^{a}$ at $\bar{D}$.
	\item At line 11, we draw the noise $n_{j}^{a}$ from a $Laplace(0,LS_{c_{j}^{a}}(\bar{D})/\epsilon_{a})$
	distribution.
\end{itemize}

To compute $LS_{c_{j}^{a}}(\bar{D})$, we need to fix the way in which attribute centroids
are computed. The following proposition gives the expression for the local sensitivity
of a centroid when the centroid is computed as the arithmetic mean of the attribute
values in the cluster. 

\begin{proposition}
	\label{prop:sens_first}Let $C^{a}=\{C_{1}^{a},\ldots,C_{p}^{a}\}$
	be a clustering of records created w.r.t. the values of attribute 
	$A^{a}$. Let $c_{j}^{a}=\frac{1}{|C_{j}^{a}|}\sum_{x\in C_{j}^{a}}x^{a}$
	, for $j=1,\ldots,p,$ be the centroid associated with cluster $C_{j}^{a}$,
	where $x^{a}$ is the value of record $x$ for attribute $A^{a}$.
	The local sensitivity for each centroid, $LS_{c_{j}^{a}}(D)$, is
	\[
	\frac{\max\{\max A^{a}-\min_{x\in C_{j}^{a}}x^{a},\,\max_{x\in C_{j}^{a}}x^{a}-\min A^{a}\}}{|C_{j}^{a}|}
	\]
	where $\max A^{a}$ and $\min A^a$ are, respectively, the maximum 
	and the minimum of the domain of $A^{a}$.
\end{proposition}

\begin{proof}
	The local sensitivity of $c_{j}^{a}$ measures the greatest change
	in $c_{j}^{a}$ that can occur as a consequence of 
	a change in one of the records in $C_{j}^{a}$.
	Since the centroid is computed as the arithmetic mean of the values of attribute $A^{a}$
	for all records in the cluster, the largest change in $c_{j}^{a}$ happens when the
	change in the record for such attribute is greatest. 
	
	Specifically, the maximum change in a record $x\in C_{j}^{a}$ is reached when 
	changing the record value $x^{a}$ by one of the extremes of the domain: either $\max A^{a}$
	or $\min A^{a}$. Thus, we can express the local sensitivity as
	\[
	LS_{c_{j}^{a}}(D)=\frac{\max_{x\in C_{j}^{a}}\{\max A^{a}-x^{a},\,x^{a}-\min A^{a}\}}{|C_{j}^{a}|},
	\]
	which is equivalent to
	\[
	\frac{\max\{\max A^{a}-\min_{x\in C_{j}^{a}}x^{a},\,\max_{x\in C_{j}^{a}}x^{a}-\min A^{a}\}}{|C_{j}^{a}|}.
	\]
\end{proof}

\subsubsection{iDP Data Sets Using Cluster-based Local Sensitivity\label{subsec:third}}

As discussed in Section~\ref{sec:process}, 
to compute the global and local sensitivities we need the domains of attributes  
to be bounded. However, for some attributes 
(\emph{e.g.}, income), there may not be a natural limit and we may need
to artificially bound the attributes to apply the algorithm. 
Bounding the attribute domain to the maximum and minimum values in the data set
is problematic, because those maximum and
minimum values correspond to specific individuals and DP is designed precisely
to hide information about 
any individual. Alternatively, we can
fix the bounds 
in a way that is independent from the actual data set.
However, if we fix the bounds without taking the original data set into account,
the empirical distribution of the data set 
is likely to be more compact than the data-independent bounds. 
This problem
also arises with attribute domains that are bounded 
but include a few outliers within their bounds.
In either case, we will get local sensitivities much larger 
than the ones that correspond to the actual data,
which leads to adding much noise and hence to poor accuracy.
Notice that the approach proposed
in the previous section to generate iDP data sets may also have the
same outlier-related issues, because it uses the attribute bounds to 
compute the local sensitivity (see Proposition~\ref{prop:sens_first}).

In this section we tackle these problems and improve the generation of iDP data sets
by making the local sensitivity of a centroid depend only on the values
within the corresponding cluster. This has two important advantages: 
(i) we avoid the shortcoming of unbounded
attributes, and (ii) we
reduce the local sensitivity (and thus the amount of noise required) 
when the attribute values in the cluster do not span the entire domain.

To make the local sensitivity of a centroid depend only on the values
of the associated cluster, we apply the following pre-processing to
the data set before using the method described in Section~\ref{subsec:first}.
For each attribute $A^{a}$ and cluster $C_{j}^{a}$, we select one
individual among those with the smallest value for $A^{a}$ in $C_{j}^{a}$ and
replace its value by the second smallest value of $A^{a}$ in $C_{j}^{a}$;
similarly, we select
one individual among those with the largest value for $A^{a}$ in $C_{j}^{a}$
and replace its value by the second largest value of $A^{a}$ in $C_{j}^{a}$.
It is important to note that our definition of second smallest (second largest) does not necessarily
imply that it is a different value from the smallest (largest); if 
there are two or more individuals that have the smallest (largest) value,
then the second smallest (largest) value is the same as the smallest (largest) value.
For example in a cluster $\{3, 3, 3, 4, 5, 6, 6\}$, we take 3 as the second smallest
value (rather than 4) and we take 6 as the second smallest value (rather than 5).
More formally, the proposed replacements are:
\begin{itemize}
	\item Let $I_{min_{j}^{a}}=\arg\min_{I_{i}\in C_{j}^{a}}\{x_{i}^{a}\}$
	be one individual with the smallest value for $A^{a}$ in $C_{j}^{a}$, and let
	$I_{min_{j}^{'a}}=\arg\min_{I_{i}\in C_{j}^{a}\setminus{I_{min_{j}^{a}}}}\{x_{i}^{a}\}$
	be one individual with the second smallest value in $C_{j}^{a}$.
	We replace the value of $I_{min_{j}^{a}}$ for $A^{a}$ by the value of $I_{min_{j}^{'a}}$, that is,
	$x_{min_{j}^{a}}^{a}=x_{min_{j}^{'a}}^{a}$.
	\item Let $I_{max_{j}^{a}}=\arg\max_{I_{i}\in C_{j}^{a}}\{x_{i}^{a}\}$
	be one individual with the largest value for $A^{a}$ in $C_{j}^{a}$, and let
	$I_{max_{j}^{'a}}=\arg\max_{I_{i}\in C_{j}^{a}\setminus{I_{max_{j}^{a}}}}\{x_{i}^{a}\}$
	be one individual with the second largest value in $C_{j}^{a}$. We
	replace the value of $I_{max_{j}^{a}}$ for $A^{a}$ by the value of $I_{max_{j}^{'a}}$, that is,
	$x_{max_{j}^{a}}^{a}=x_{max_{j}^{'a}}^{a}$.
\end{itemize}

Let us call $D'$ the data set that results from applying this
pre-processing step to $D$. The purpose of the pre-processing is to
make sure that modification of a single record of $D$ 
keeps the values of cluster $C_{j}^{a}$ in $D'$ within the range
$[x_{min_{j}^{a}}^{a},x_{max_{j}^{a}}^{a}].$ Indeed, if we modify
a record to a value smaller than $x_{min_{j}^{a}}^{a}$, resp. larger than $x_{max_{j}^{a}}^{a}$
(which causes the modified value to become the smallest, 
resp. the largest value in $C_{j}^{a}$) 
the pre-processing step will automatically replace the modified
value by $x_{min_{j}^{a}}^{a}$, resp.
$x_{max_{j}^{a}}^{a}$ 
(that is, the former smallest value, resp. largest value, which is 
now the second smallest value, resp. second largest value).

Composing the pre-processing step with the microaggregation
can be viewed as an alternative microaggregation algorithm, which we expect to
be substantially less sensitive to changes of the records in $D$. 
To attain iDP
we need to adjust the noise to the sensitivity of this alternative
microaggregation algorithm. Such a sensitivity is computed in the following proposition.

\begin{proposition}
	\label{prop:sens_second}Let $C^{a}=\{C_{1}^{a},\ldots,C_{p}^{a}\}$
	be a clustering of records created w.r.t. the values of attribute 
	$A^{a}$ 
	and assume each cluster contains at least three values (not necessarily
	different). 
	Let $P^{a}=\{P_{1}^{a},\ldots,P_{p}^{a}\}$ be the clustering that results
	from replacing the largest and the smallest values of attribute $A^{a}$ in
	each cluster by the second largest and the second smallest value in the cluster
	(that is, the clusters that result from applying the previously 
	described pre-processing step to $C^a$).
	Let $p_{j}^{a}=\frac{1}{|P_{j}^{a}|}\sum_{x\in P_{j}^{a}}x^{a}$
	, for $j=1,\ldots,p,$ be the centroid associated with 
	cluster $P_{j}^{a}$,
	where $x^{a}$ is the value of the record $x$ for the attribute $A^{a}$.
	The local sensitivity for each centroid, $LS_{p_{j}^{a}}(D)$, is
	
	\[	\max\{|x_{max^{a}_j}-x_{min^{'a}_j}|+|x_{min^{''a}_j}-x_{min^{'a}_j}|+
	|x_{max^a_j} - x_{max^{'a}_j}|, \]
	\begin{equation}	
		|x_{min^a_j}-x_{max^{'a}_j}|+|x_{max^{''a}_j}-x_{max^{'a}_j}|+
		|x_{min^a_j} - x_{min^{'a}_j}|\}
		\label{eq:local_sens3}
	\end{equation}
	divided by $|P_{j}^{a}|$, where $x_{min^{''a}_j}$ and $x_{max^{''a}_j}$
	are, respectively, the third smallest and the third largest values in 
	the cluster.
\end{proposition}

\begin{proof}
	Let us assume the maximum possible change of a record.
	This occurs when the smallest value becomes greater than the largest value so far,
	or when the largest value becomes smaller than the smallest value so far.
	Let us begin with the first case.
	
	Assume $x_{min^a_j}$ changes to a value 
	greater than the largest value $x_{max^a_j}$.
	Initially, the pre-processed value was $x_{min^{'a}_j}$ and it becomes
	$x_{max^{a}_j}$ after the change (as $x_{max^{a}_j}$ is the second largest value of the cluster
	in the modified data set). Additionally,
	after pre-processing, the former second smallest value
	$x_{min^{'a}_j}$ is changed to the third smallest value $x_{min^{''a}_j}$ and $x_{max^{'a}_j}$
	becomes $x_{max^a_j}$. Hence, the change in the sum of pre-processed cluster values is
	\begin{equation}
		\label{part1}
		|x_{max^{a}_j}-x_{min^{'a}_j}|+|x_{min^{''a}_j}-x_{min^{'a}_j}|+|x_{max^a_j} - x_{max^{'a}_j}|.
	\end{equation}
	
	Symmetrically, if $x_{max^a_j}$ changes to a value less than
	the smallest value $x_{min^a_j}$, the change in the sum of pre-processed cluster values is
	\begin{equation}
		\label{part2}
		|x_{min^a_j}-x_{max^{'a}_j}|+|x_{max^{''a}_j}-x_{max^{'a}_j}|+|x_{min^a_j} - x_{min^{'a}_j}|.
	\end{equation}
	
	Let us now consider changes in the second smallest or second largest values,
	which are also relevant for the range of the pre-processed cluster.
	If the second smallest record changes maximally, that is, becomes
	greater than the largest value so far, the pre-processed cluster
	will have values in the range $[x_{min^{''a}_j},x_{max^a_j}]$. Hence,
	the change in the sum of pre-processed cluster 
	values is the same as in Expression (\ref{part1}). Symmetrically,
	if the second largest value changes maximally, the change in the sum of values is 
	the same as in Expression (\ref{part2}).
	
	Thus, the maximum change in the sum of pre-processed cluster values is the maximum of Expressions (\ref{part1}) 
	and (\ref{part2}). To obtain the sensitivity, we must divide
	by the cardinality $|P_{j}^{a}|$ of the cluster. This concludes
	the proof.
\end{proof}

In this way, if we add the pre-processing step
to the approach proposed in Section~\ref{subsec:first}, 
the local
sensitivity of each centroid can be computed from 
the values in the corresponding cluster.
Hence, we do not need the attributes to be bounded
and we obtain smaller local sensitivities
--because the cluster value range is usually narrower than
the domain value range--.

By using the approach described in Section~\ref{subsec:first} on 
the pre-processed data set $(D')$ and computing the local sensitivity
as specified in Expression (\ref{eq:local_sens3}), we generate
an $\epsilon$-iDP data set. 

\section{Experimental Evaluation\label{sec:experiments}}

This section reports the empirical evaluation of the two microaggregation strategies 
we propose for generating iDP data sets.
Two well-known data sets have been used as evaluation data:

\begin{itemize}
	\item \emph{Census}: This is a test data set extracted from the 
	U.S. Census Bureau 1995 Current Population Survey~\cite{Brand}. 
	It contains 1,080 records with numerical attributes and it has been widely used for evaluating privacy-preserving methods~\cite{Soria2014,Sanchez2016,Lasz05}. 
	The following (unbounded) attributes about finance and taxes have been used in our experiments: AFNLWGT, AGI, EMCONTRB, FEDTAX, STATETAX, TAXINC, POTHVAL, INTVAL, and FICA. 
	\item \emph{Wine Quality}: This is a double data set for 
	classification/regression tasks found in UCI~\cite{Wine}.
	There is a data set  
	related to red wines and another related to white wines.
	We used the white wine dataset, 
	which contains significantly more instances than
	the red wine data set (4,898 vs. 1,600). 
	Attributes are discrete 
	or numerical, and they describe 
	physicochemical properties of wine. 
\end{itemize}

All the attributes of the two data sets we used represent 
non-negative numerical magnitudes and most of them are unbounded.
To compute the (global and local) sensitivities, we fixed the attribute
domains as 
\[ [0, \ldots, (\alpha \times max\_attr.\_value\_in\_data set)],\]
and took $\alpha\in \{1.5,3.0\}$. By varying $\alpha$ we tested the influence of the
size of the attribute domains on the sensitivity that applies for each method.
Since the Laplace distribution takes values in the range $(-\infty, +\infty)$, 
for consistency we bounded noise-added outputs to the domain ranges defined above.

Two well-differentiated evaluation experiments were carried out. 
In the first one, we measured the information loss 
incurred by our methods w.r.t. that of standard DP with general metrics. 
In the second experiment, we assessed the utility retained by the 
masked data produced by our methods when employed in machine learning tasks.

\subsection{Information Loss Evaluation}

Information loss refers to the difference between the masked and the original data, that is,
to the harm inflicted by masking to the analytical accuracy of the original data.
The Sum of Squared Errors (SSE) is a standard measure of information 
loss~\cite{SDC}, and it is 
defined as the sum of squares of distances between original and masked records:

\begin{equation}
	SSE = \sum dist(r_i, (r_i)')^2 ,
\end{equation}

where $r_i$ is the $i$-th record in the original data set and $(r_i)'$ represents its masked version. 
The mean SSE, calculated by dividing the SSE by the number of records $n$, is 
usually more informative 
because it does not depend on the cardinality of the data set.

To compute SSE, we need a distance between records:

\[	d((x_1,...,x_m), (y_1,...,y_m))= \]
\begin{equation}
	=(1/m)\sqrt{(d_1(x_1,y_1)/\sigma_1^2)^2 + \ldots + (d_m(x_m,y_m)/\sigma_m^2)^2},
\end{equation}

where $d_j(\cdot,\cdot)$ is the distance between values of the $j$-th attribute,
$\sigma_j^2$ is the sample variance of the $j$-th attribute in the original data set and
$m$ is the number of attributes.

For differential privacy, we have considered $\epsilon$ values 0.01, 0.1 and 1.0, 
which cover the range of reasonably safe values~\cite{Dwork2011}.
In all the DP/iDP methods discussed/proposed in this paper, the sensitivity of the 
centroids obtained after the microaggregation step depends on the parameter $k$ that 
establishes the minimum number of records in each cluster; 
that is, it defines the cardinality $|C|$ of the cluster associated with the centroid.
In order to evaluate the influence of the microaggregation step, 
and considering the cardinality of the
\emph{Census} and \emph{Wine Quality} data sets, we have taken parameter $k$ 
between 3 and 100 for the former data set and between 3 and 400 for the latter;
notice that $k=3$ is the minimum microaggregation level supported by the method
using cluster-based local sensitivity.

To put in context the results obtained with our methods 
--iDP using microaggregation and local sensitivity (\emph{iDP-LS}) and
iDP using microaggregation and cluster-based local sensitivity (\emph{iDP-CBLS})--, 
we compared them with 
the following alternatives:
\begin{itemize} 
	\item Plain Laplace noise addition for $\epsilon$-differential privacy (\emph{DP}) 
	with no prior microaggregation, 
	as described in Section~\ref{subsec:dp}. 
	This is the naive mechanism to produce differentially private data sets. 
	We used it as an upper bound for the information loss.
	\item Differential privacy using univariate microaggregation (\emph{DP-UM}), as 
	described in Section~\ref{subsec:ir_dp}. 
	This mechanism uses the prior microaggregation step to reduce the sensitivity, 
	even though the noise applied to the 
	centroids needs to be adjusted to the global sensitivity.
\end{itemize}

For each data set, each above-mentioned method and each 
method parameterization (in terms of $\alpha$, $k$ and $\epsilon$ values), we 
made 10 runs and computed the mean SSE over them. The result is depicted in
Figures~\ref{fig_SSE_CASC} and~\ref{fig_SSE_Wine}.
Notice that the plain Laplace noise addition method (\emph{DP}) is displayed as a 
horizontal line, because it entails no microaggregation and is hence independent of $k$.
Also, since SSE values are quite diverse among the different methods, 
a $\log_{10}$ scale has been used in the ordinates.

\begin{figure}
	\begin{center}
		\includegraphics[scale=0.34]{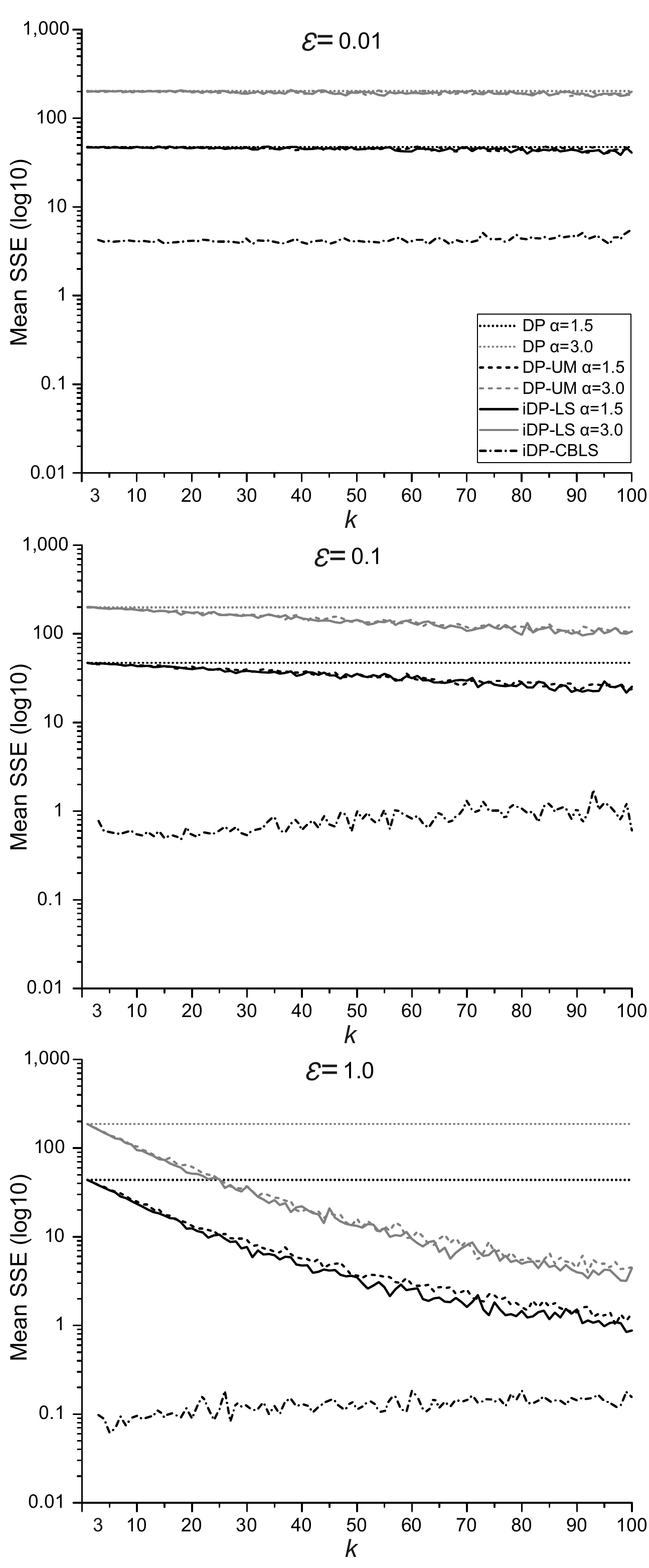}
		\caption{\emph{Census} data set: mean SSE for the proposed methods (\emph{iDP-LS} and \emph{iDP-CBLS}) and baselines (\emph{DP-UM} and \emph{DP})
			with $\epsilon=0.01$ (top), $\epsilon=0.1$ (center) and $\epsilon=1.0$ (bottom), and 
			$\alpha\in \{1.5,3.0\}$,
			for microaggregation parameter $k$ from 3 to 100}
		\label{fig_SSE_CASC}
	\end{center}
\end{figure}

\begin{figure}
	\begin{center}
		\includegraphics[scale=0.34]{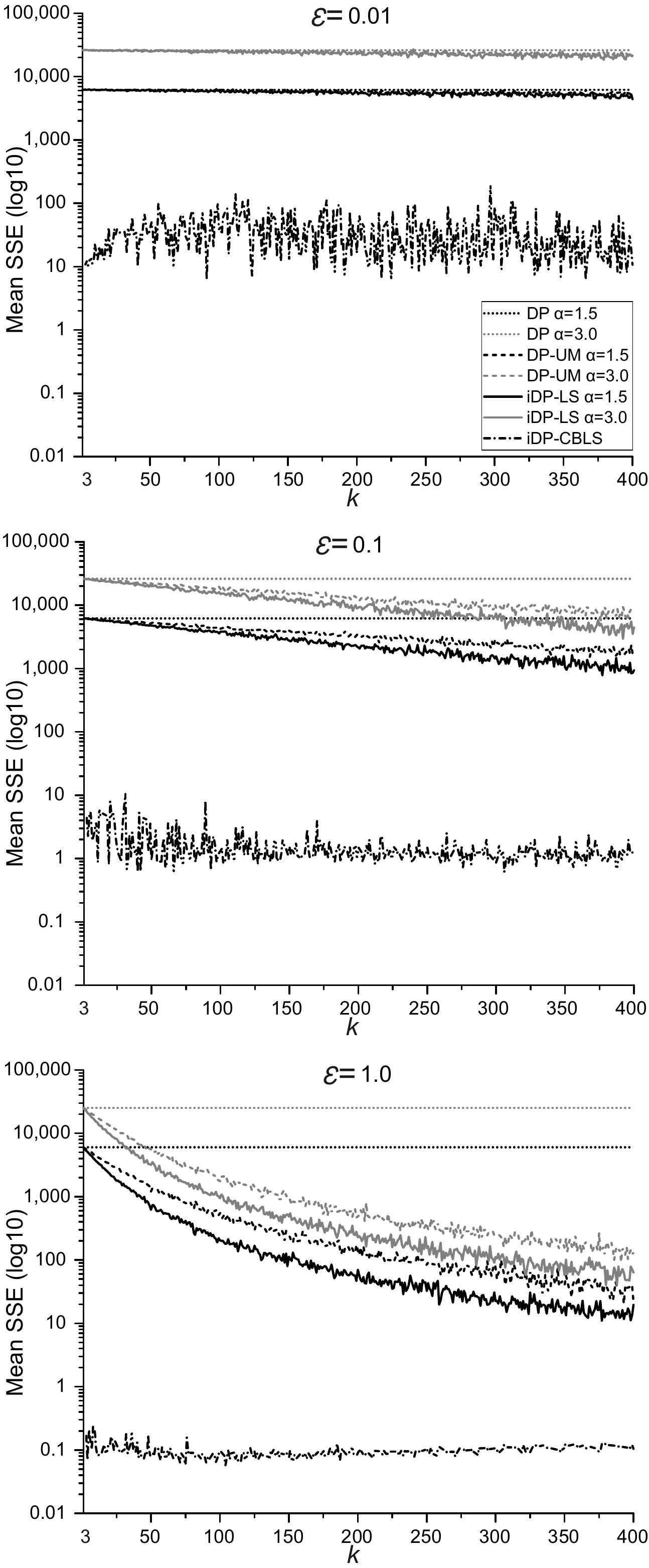}
		\caption{\emph{Wine Quality} data set: mean SSE for the proposed methods (\emph{iDP-LS} and \emph{iDP-CBLS}) and baselines (\emph{DP-UM} and \emph{DP})
			with $\epsilon=0.01$ (top), $\epsilon=0.1$ (center) and $\epsilon=1.0$ (bottom), and $\alpha\in \{1.5,3.0\}$,
			for microaggregation parameter $k$ from 3 to 400}
		\label{fig_SSE_Wine}
	\end{center}
\end{figure}

The results obtained for the two data sets show that:
\begin{itemize}
	\item Plain Laplace noise addition ({\em DP}), with no prior microaggregation, results in the highest SSE due to the large global sensitivity.
	In fact, SSE barely decreases when moving from $\epsilon=0.01$ to $\epsilon=1.0$ and, thus, we may consider the 
	masked data nearly random in all cases.
	\item Both {\em DP-UM} and {\em iDP-LS} yield similar results, with SSE decreasing 
	as the microaggregation level ($k$) increases. This shows the benefits
	of the prior microaggregation step to reduce the sensitivity. {\em iDP-LS} 
	achieves lower SSE thanks to its ability to use the local sensitivity
	instead of the global sensitivity employed by {\em DP-UM}. The differences 
	between the two methods are larger for higher values of $\epsilon$,
	but stay proportional when varying the boundaries of the attribute domains ($\alpha$). 
	In this respect, even though {\em iDP-LS} uses the local sensitivity,
	it still depends on the domain boundaries, as stated in Proposition~\ref{prop:sens_first}. 
	\item For the three previous methods, SSE values 
	decrease as the attribute domains get smaller (smaller $\alpha$). 
	This illustrates that large (or even unbounded)
	domains severely deteriorate the utility of the masked data when 
	using the straightforward approach to generate differentially private data sets. 
	\item Our most sophisticated strategy ({\em iDP-CBLS}) is able to improve the accuracy of 
	the previous methods by several orders of magnitude.
	In fact, its SSE for $\epsilon=0.01$ is similar to the ones 
	of {\em DP-UM} and {\em iDP-LS} for $\epsilon=1.0$ and large $k$, which shows that {\em iDP-CBLS}
	can accommodate stronger privacy requirements (smaller $\epsilon$).
	Moreover, since the sensitivity calculation is based on clusters rather than on attribute domains (see Proposition~\ref{prop:sens_second}), it is not affected by
	(large) domains. Also, unlike the former methods, {\em iDP-CBLS} 
	requires very small microaggregation levels ($k$ among 5-15) to obtain optimal results.
	Larger microaggregation levels produce a small SSE increase 
	because the distortion caused by microaggregation is greater (for large $k$) 
	than the benefits resulting from reducing the sensitivity 
	(the cluster-based sensitivity is already quite small).  
\end{itemize}

\subsection{Evaluation in Data Classification}

The second set of experiments aims at evaluating the utility 
retained by the masked outcomes in a specific machine learning task, namely data classification. This scenario is especially relevant because DP has also been adopted as the \emph{de facto} standard for privacy protection in machine learning and, like in data releases, researchers struggle to reconcile meaningful DP guarantees with usable model accuracy~\cite{review}.

We built a classification model 
using masked training data, and we compared its 
classification accuracy with that of a model built on original training data.
The same original data were used in both cases as evaluation data to 
measure the classification accuracy.
For both the {\em Census} and the {\em Wine Quality} data sets, 
we used the first 66\% of (masked, resp. original) records for training 
and the rest of (original) records for evaluation. 
As in the former experiments, we report here the average result of 
10 runs for the same parameter values ($\epsilon$, $k$ and $\alpha$) and for the
two methods we propose: {\em iDP-LS} and {\em iDP-CBLS}.

The classifier we chose is Random Forest, which is fast and easy 
to implement, produces highly accurate predictions and can handle a 
very large number of input variables without overfitting~\cite{Biau12}.

To measure the classification accuracy, we focused on the F-measure of 
the class attribute,
which is the harmonic mean between precision 
and recall. 

\subsubsection{\emph{Census} data set}

Since {\em Census} is a general purpose data set (not specifically 
aimed at classification), we had to adopt ERNVAL 
(business or farm net earnings in the year of reference) as class attribute
by categorizing it into ``$\leq$30K'' and ``>30K'' if the balance 
is less/equal than 30,000 and greater than 30,000, respectively. 
Thus, the classification objective is to predict whether an individual obtains 
earnings below or above 30,000.

Fig.~\ref{fig:CASC} depicts the F-measure for both classes with original 
training data and with the masked training data produced by our methods 
for the different parameters.
In general, the results are coherent with the SSE figures 
reported above. On the one hand, the results of {\em iDP-LS} improve proportionally 
to the microaggregation level $k$ and
are better for narrower attribute domains (which result in a smaller sensitivity)
with $\epsilon>0.01$.
For such values of $\epsilon$, F-measures remain quite stable from $k=50$ onwards, 
and the best results are around 20\%, 10\% and 5\% worse
for $\epsilon=0.01$, $\epsilon=0.1$ and $\epsilon=1.0$, respectively, 
than the upper bound defined by the original training data. 
On the other hand, {\em iDP-CBLS} provides much better results, that 
for $\epsilon=1.0$ are identical to the upper bound,
and for $\epsilon=0.1$ and $\epsilon=0.01$
only 3\% and 10\% worse, respectively. 
In this scenario, the utility retained by {\em iDP-CBLS} for $\epsilon=0.01$ 
is similar to that of {\em iDP-LS} for $\epsilon=0.1$.
As we also observed in the former experiments, F-measures for {\em iDP-CBLS} degrade 
for large $k$ because the distortion 
added by the microaggregation step becomes comparatively larger 
than the reduction of sensitivity (and thus than the reduction 
of the noise, which is already very small).

\begin{figure}
	\begin{center}
		\includegraphics[width=\linewidth]{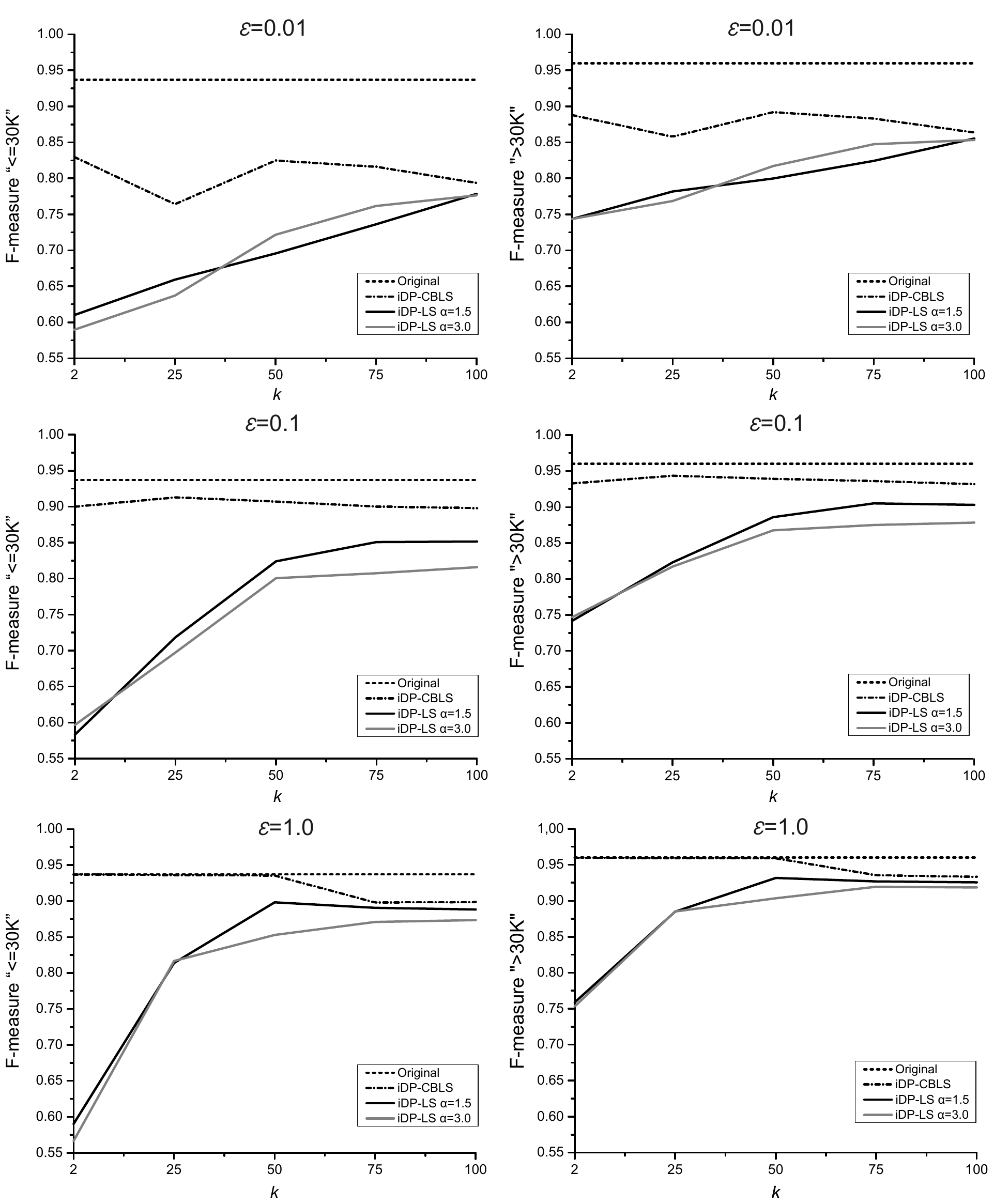}
		\caption{\emph{Census} data set: F-measures for classes ``$\leq$30K'' (left) and ``>30K'' (right) with original training data and masked training data for different parameters}
		\label{fig:CASC}
	\end{center}
\end{figure}

\subsubsection{\emph{Wine Quality} data set}

In the \emph{Wine Quality} data set, the class attribute is QUALITY, 
ranging from 0 to 10, where 0 is the worst quality score and 10 the best.
To simplify the classification, we considered two classes: ``Not Excellent'' 
and ``Excellent'', which correspond to QUALITY values $\leq$ 6 or $>$ 6 respectively. 

Fig.~\ref{fig:WINE} depicts the F-measures for these two classes with 
the same methods and parameters as above. 
The tendencies observed in the results are similar to those for \emph{Census}, 
but with some differences. On the one hand, the F-measures for the two 
classes are significantly different
because the classes are unbalanced: ``Not Excellent'' has 3.5 times 
more records than ``Excellent''; thus, the latter class is more difficult 
to classify due to the smaller
amount of training data, and results in lower F-measures. On the 
other hand, {\em iDP-CBLS} for $\epsilon=0.1$ and $\epsilon=1.0$ was able
not only to reach the upper bound, but also to
slightly improve it for some values of $k$. This phenomenon 
can be explained by the fact that adding a small amount of noise to the training data, 
as is the case for {\em iDP-CBLS}, may improve the 
classification accuracy as long as the distribution of the data remains similar~\cite{vandermaaten13}. Thus, we can observe that {\em iDP-CBLS} was not only able to maintain a low 
SSE, but
the (small) noise it added had even a positive effect on data classification in some cases.

\begin{figure}
	\begin{center}
		\includegraphics[width=\linewidth]{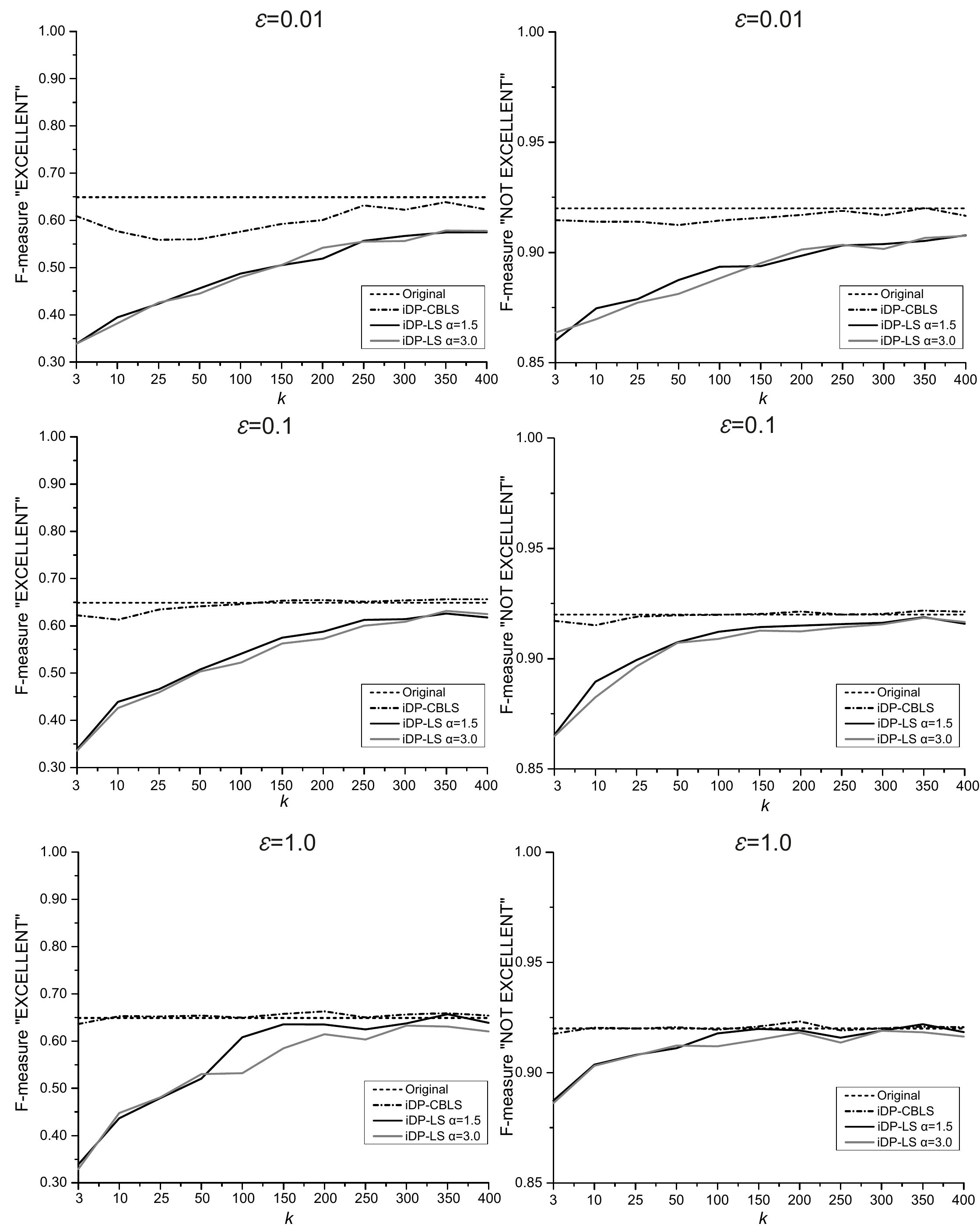}
		\caption{\emph{Wine Quality} data set: F-measures for classes ``Excellent'' (left) and ``Not Excellent'' (right) with original training data and masked training data for different parameters}
		\label{fig:WINE}
	\end{center}
\end{figure}

\section{Conclusions and Future Research\label{sec:Conclusions}}

The strong privacy guarantees of DP only hold for small values of 
$\epsilon$, which usually result in
limited data accuracy in non-interactive settings, such as data releases and machine learning~\cite{limits}. 

To tackle this problem, in this paper we have leveraged the (potential) advantages of iDP regarding data utility
with the sensitivity (and noise) reduction enabled by microaggregation-based DP masking.
Two microaggregation strategies have been proposed, of which the second
(that computes the local sensitivity based on clusters) is the  
one that brings the greatest utility benefits. 
In fact, with {\em iDP-CBLS} we were able to improve the utility of the masked 
outcomes for a given $\epsilon$ by several orders of magnitude in comparison with plain DP data sets.
As shown in the experiments, for a given $\epsilon$, the 
utility preserved by {\em iDP-CBLS} rivals
the one preserved by DP data sets with 10-100 times larger $\epsilon$ values.
Therefore, our method is able to 
reconcile the small values of $\epsilon$ (say $\epsilon < 0.1$) that 
are needed for DP to offer real privacy, with
the data accuracy that is needed for data releases to be useful.
Also, for larger $\epsilon$ values, {\em iDP-CBLS} was not only able 
to keep the noise low, but
also to reach the accuracy upper bound 
for classification tasks (corresponding to the case 
in which original data are used for training). 
These results are significantly better than those obtained when applying standard DP in machine learning tasks, even with much larger $\epsilon$~\cite{review}.

As future work, we plan to design iDP-enforcing mechanisms (other than noise addition)
for non-numerical discrete attributes, which are commonly found in microdata sets.
We also plan to explore the use of multivariate microaggregation instead
of individual-ranking microaggregation. Although multivariate microaggregation
might cause more information loss, it might also guarantee that the protection
offered by {\em iDP-CBLS} can be no less than the one offered by
$k$-anonymity, no matter the value of $\epsilon$ chosen.

\section*{Acknowledgments and disclaimer}

This research was funded by the European Commission (projects H2020-871042 ``SoBigData++'' and H2020-101006879 ``MobiDataLab''), the Government of Catalonia (ICREA Acad\`emia Prizes to J.Domingo-Ferrer and to D. S\'anchez), MCIN/AEI/ 10.13039/501100011033 and ``ERDF A way of making Europe'' under grant PID2021-123637NB-I00 ``CURLING'' and INCIBE-NextGenerationEU (project ``HERMES'' and 
INCIBE-URV cybersecurity chair).
The authors are with the
UNESCO Chair in Data Privacy, but the views in this paper
are their own and are not necessarily shared by UNESCO or the 
above mentioned funding agencies.



\begin{thebibliography}{}

\bibitem{SDC}
A. Hundepool, J. Domingo-Ferrer, L. Franconi, S. Giessing, E.S. Nordholt, K. Spicer, P.P. de~Wolf, Statistical Disclosure Control, Wiley, 2012.

\bibitem{Soria2015big}
J. Soria-Comas, J. Domingo-Ferrer, 2015, Big data privacy: challenges to privacy principles and models, Data Science and Engineering 1(1)1--8.

\bibitem{DworkMNS06}
C. Dwork, F. McSherry, K. Nissim, A.D. Smith, Calibrating noise to sensitivity in private data analysis, in: 3rd Theory of Cryptography Conference, Springer, 2006, pp. 265--284.

\bibitem{Samarati2001}
P. Samarati, 2001, Protecting respondents' identities in microdata release, IEEE Transactions on Knowledge and Data Engineering 13(6)1010--1027.

\bibitem{Machanavajjhala2007}
A. Machanavajjhala, D. Kifer, J. Gehrke, M. Venkitasubramaniam, 2007, $l$-Diversity: privacy beyond $k$-anonymity, ACM Transactions on Knowledge Discovery from Data 1(1).

\bibitem{Li2007}
N. Li, T. Li, S. Venkatasubramanian, $t$-closeness: privacy beyond $k$-anonymity and $l$-diversity, in: Proceedings of the 23rd IEEE International Conference on Data Engineering (ICDE), 2007, pp. 106--115.

\bibitem{Zhang2014}
J. Zhang, G. Cormode, C.M. Procopiuc, D. Srivastava, X. Xiao, Privbayes: private data release via Bayesian networks, in: Proceedings of the 2014 ACM SIGMOD International Conference on Management of Data, 2014, pp. 1423--1434.

\bibitem{Xiao2010}
Y. Xiao, L. Xiong, C. Yuan, Differentially private data release through multidimensional partitioning, in: Secure Data Management (SDM), Springer, 2010, pp. 150--168.

\bibitem{Soria2014}
J. Soria-Comas, J. Domingo-Ferrer, D. S\'{a}nchez, S. Mart\'{\i}nez, 2014, Enhancing data utility in differential privacy via microaggregation-based $k$-anonymity, The VLDB Journal 23(5)771--794.

\bibitem{Sanchez2016}
D. S\'anchez, J. Domingo-Ferrer, S. Mart\'inez, J. Soria-Comas, 2016, Utility-preserving differentially private data releases via individual ranking microaggregation, Information Fusion 30(1)1--14.

\bibitem{Chu2023}
Z. Chu, J. He, J. Li, Q. Wang, X. Zhang, N. Zhu, 2023, SSKM\_DP: Differential privacy data publishing method via SFLA-Kohonen network, Applied Sciences 13(6)3823.

 \bibitem{limits}
J. Domingo-Ferrer, D. S\'anchez, A. Blanco-Justicia, 2021, 
The limits of differential privacy (and its misuse in data release and machine learning), Communications of the ACM 64(7)33--35.

\bibitem{Abowd}
J.M. Abowd, M.B. Hawes, Confidentiality protection in the 2020 U.S. Census of Population and Housing, in: Annual Review of Statistics and Its Application, 2023, pp. 119--144.

\bibitem{USCB22} U.S. Census Bureau, Privacy-loss Budget Allocation, 2022, \url{https://www2.census. gov/programs-surveys/decennial/2020/program-management/data-product-planning/ 2010-demonstration-data-products/02-Demographic_and_Housing_Characteristics/ 2022-03-16_Summary_File/2022-03-16_Privacy-Loss_Budget_Allocations.pdf} (accessed 6 October 2023).

\bibitem{Kenny21}
C. Kenny, et al, 2021, The use of differential privacy for census data and its impact on redistricting: The case of the 2020 U.S. Census, Science Advances 7(41)eabk3283.

\bibitem{Dwork19}
C. Dwork, N. Kohli, M. Mulligan, 2019, Differential privacy in practice: expose your epsilons!, Journal of Privacy and Confidentiality 9(2).

\bibitem{iDP}
J. Soria-Comas, J. Domingo-Ferrer, D. S\'anchez, D. Meg\'ias, 2017, Individual differential privacy: a utility-preserving formulation of differential	privacy guarantees, IEEE Transactions on Information Forensics and Security 12(6)1418--1429.

\bibitem{Domi02}
J. Domingo-Ferrer, J.M. Mateo-Sanz, 2002, Practical data-oriented
microaggregation for statistical disclosure control, IEEE Transactions on Knowledge and Data Engineering 14(1)189--201.

\bibitem{Dwork2011}
C. Dwork, 2011, A firm foundation for private data analysis, Communications of the ACM 54(1)86--95.

\bibitem{Dwork2006}
C. Dwork, Differential privacy, in: Automata, Languages and Programming (ICALP), Springer, 2006, pp. 1--12.

\bibitem{DworkKMMN06}
C. Dwork, K. Kenthapadi, F. McSherry, I. Mironov, N. Naor, Our data, ourselves: privacy via distributed noise generation, in: Advances in Cryptology - {EUROCRYPT}, Springer, 2006, pp. 486--503.

\bibitem{DworkR16}
C. Dwork, G.N. Rothblum, Concentrated differential privacy, CoRR abs/1603.01887, 2016, \url{http://arxiv.org/abs/1603.01887}

\bibitem{BunandSteinke2016}
M. Bun, T. Steinke, Concentrated differential privacy: Simplifications, extensions, and lower bounds, in: Theory of Cryptography Conference, Springer, 2016, pp. 635--658.

\bibitem{Mironov2017}
I. Mironov, R\'enyi differential privacy, in: IEEE 30th Computer Security Foundations Symposium (CSF), IEEE, 2017, pp. 263--275.

\bibitem{Nissim2007}
K. Nissim, S. Raskhodnikova, A. Smith, Smooth sensitivity and sampling in private data analysis, in: Proceedings of the 39th Annual ACM Symposium on Theory of Computing, 2007, pp. 75--84.

\bibitem{Soria2017mdai}
J. Soria-Comas, J. Domingo-Ferrer, Differentially private data sets based on microaggregation and record perturbation, in: Modeling Decision for Artificial Intelligence (MDAI), Springer, 2017, pp. 119--131.

\bibitem{Brand}
R. Brand, J. Domingo-Ferrer, J.M. Mateo-Sanz, Reference data sets to test and compare SDC methods for protection of numerical microdata, CASC Project Deliverable, 2002, \url{https://research.cbs.nl/casc/CASCrefmicrodata.pdf/} (accessed 8 October 2023).

\bibitem{Lasz05}
M. Laszlo, S. Mukherjee, 2005, Minimum spanning tree partitioning algorithm for microaggregation, IEEE Transactions on Knowledge and Data Engineering 17(7)902--911.

\bibitem{Wine}
Wine Quality Data Set, UCI Machine Learning Repository, 2009,
\url{https://archive.ics.uci.edu/ml/datasets/wine+quality/} (accessed 6 October 2023).

\bibitem{review}
A. Blanco-Justicia, D. S\'anchez, J. Domingo-Ferrer, K. Muralidhar, 2023, A critical review on the use (and misuse) of differential privacy in machine learning, ACM Computing Surveys 55(8)1--26.

\bibitem{Biau12}
G. Biau, 2012, Analysis of a random forests model, Journal of Machine Learning Research 13(1)1063--1095.

\bibitem{vandermaaten13}
L. Maaten, M. Chen, S. Tyree, K. Weinberger, Learning with marginalized corrupted features, in: Proceedings of 30th International Conference on Machine Learning (ICML), 2013, pp. 410--418.
	
\end{thebibliography}
\end{document}